\documentclass[twoside,12pt]{article}
\usepackage{CJK}
\usepackage{indentfirst, color}
\usepackage{bm}
\usepackage{graphicx}
\usepackage{epsfig}
\usepackage{amsmath}
\usepackage{amsfonts}
\usepackage{amssymb}
\usepackage{amsthm}
\usepackage[superscript]{cite}
\usepackage{latexsym}

\newtheorem{thm}{Theorem}

\usepackage{epsfig}
\usepackage{epstopdf}
\usepackage{pgf,fancyhdr}
\usepackage{float}
\usepackage{amsmath}

\begin{document}

\begin{CJK*}{GBK}{song}

\begin{center}
\LARGE\bf Criteria of genuine multipartite entanglement based on correlation tensors
\end{center}

\begin{center}
\rm Naihuan Jing,$^{1, 2,*}$ \  Meiming Zhang$^1$
\end{center}

\begin{center}
\begin{footnotesize} \sl
$^1$ Department of Mathematics, Shanghai University, Shanghai 200444, China  

$^2$ Department of Mathematics, North Carolina State University, Raleigh, NC 27695, USA

$^*$ Corresponding author: jing@ncsu.edu

\end{footnotesize}
\end{center}

\begin{center}
\begin{minipage}{15.5cm}
\parindent 20pt\footnotesize
We revisit the genuine multipartite entanglement by a simplified method, which only involves
the Schmidt decomposition and local unitary transformation.
We construct a local unitary equivalent class of the tri-qubit quantum state, then use the trace norm of the whole correlation tensor as a measurement to detect
genuine multipartite entanglement.
By detailed examples, we show our result can detect more genuinely entangled states.
Furthermore, we generalize the genuine multipartite entanglement criterion to tripartite higher-dimensional systems.
\end{minipage}
\end{center}

\begin{center}
\begin{minipage}{15.5cm}
\begin{minipage}[t]{2.3cm}{\bf Keywords:}\end{minipage}
\begin{minipage}[t]{13.1cm}
Genuine multipartite entanglement, Correlation tensor
\end{minipage}\par\vglue8pt

\end{minipage}
\end{center}
\section{Introduction}
Quantum entanglement [1] plays a crucial role in quantum information processing for enhancing performance over classical methods.
Multipartite entanglement, in particular, genuinely multipartite entangled (GME), is viewed as an essential resource to perform tasks such as quantum cryptography, quantum simulations, quantum networks, etc. [2-5].

For a given quantum state, how to determine whether it is entangled, how to measure the entanglement, and
how to use the entanglement property to better realize quantum information processing tasks are some of the key questions
requiring investigations. 
Numerous works have been done towards the detection of entanglement for bipartite states. However, how to quantify the genuine entanglement contained in the multipartite state is still lacking complete understanding.
To identify GME, Bell-like inequalities [6, 7, 8], various entanglement witnesses [9-14], Fisher information-based [15] and generalized concurrence for multi genuine entanglement [16-19] were proposed. 
Some entanglement criteria based on local uncertainty relations were presented to detect entangled states or bound entangled states [20-23].
It was realized that the norm of the correlation tensor can be used to study entanglement [24-31] as well as for determining partial or full separability in multipartite states [24-28] and GME in multipartite systems [29, 30, 31].
In recent years, methods for detecting bipartite entanglement based on partially transposed density matrices were also provided [32, 33].

In this paper, we would like to use local unitary equivalence to simplify the determination of GME. Our idea is to first characterize the genuine tripartite entanglement by using the Schmidt decomposition and local unitary transformation, and then use the trace norm of the correlation tensor to judge entanglement. 
The current study tries to streamline and simplify some of the recent works on GME, where only partial or submatrices of the correlation tensors are considered.
We find that by combining with local unitary equivalence nice and simple detection criteria for GME are accessible. The main idea is to cast the correlation tensor into a rectangular matrix
and then use the trace norm for our measurement of GME, which turns out to be a convenient benchmark in this regard. 

The paper is organized as follows. In Sec. 2, we derive the GME of tri-qubit quantum states in terms of the correlation tensors of the density matrices.
In Sec. 3, we generalize the GME criterion of $2\otimes2\otimes2$ quantum systems to $d\otimes d\otimes d$ quantum systems. By detailed examples, our results are seen to outperform previously published results. Finally the conclusion is summarized in Sec. 4.
\section{GME for tri-qubit quantum states}
Two bipartite states $\rho$ and $\rho'$ are said to be local unitary (LU) equivalent if there exist unitary operators $U_1$, and $U_2$ such that
\begin{eqnarray}
\rho'=(U_1\otimes U_2)\rho(U_1\otimes U_2)^\dag,
\end{eqnarray}
where $\dag$ denotes transpose and conjugate.
Let $\lambda_i$, $i=1,\cdots,d_n^2-1$, be the generators of  the special unitary Lie algebra $\mathfrak{su}(d_n)$. For any bipartite state $\rho\in\mathcal{H}_1^{d_1}\otimes\mathcal{H}_2^{d_2}$, $\rho$ can be denoted as $\rho=\frac{1}{d_1d_2}\sum_{f=0}^{d_1^2-1}\sum_{g=0}^{d_2^2-1}\langle\lambda_f\otimes\lambda_g\rangle\lambda_f\otimes\lambda_g$, where $\lambda_0$ denotes the identity operator $\mathbf{1}$ and $t_{f, g}=\langle\lambda_f\otimes\lambda_g\rangle=tr(\rho\lambda_f\otimes\lambda_g)$. Let $T=(t_{f, g})_{(d_1^2-1)\times (d_2^2-1)}$ be the matrix of the coefficients. The trace norm of a rectangular matrix $A\in\mathbb{R}^{m\times n}$ is defined as the sum of the singular values, i.e., $\|A\|_{tr}=\sum_i\sigma_i=tr\sqrt{A^\dag A}$, where $\sigma_i$, $i=1,\cdots,min(m,n)$, are the singular values of the matrix $A$ arranged in descending order.
If two bipartite states $\rho$ and $\rho'$ are LU equivalent, then the trace norm of $T$ is invariant, i.e, $\|T(\rho)\|_{tr}=\|T(\rho')\|_{tr}$ [35, 36]. Using the LU equivalence we can simplify the study of separability or entanglement of the density matrix. 

A tripartite quantum state $\rho\in\mathcal{H}_1^2\otimes\mathcal{H}_2^2\otimes\mathcal{H}_3^2$ can be represented as:
\begin{eqnarray}
\rho=\frac{1}{8}\sum_{f,g,h=0}^{3}\langle\lambda_f\otimes\lambda_g\otimes\lambda_h\rangle\lambda_f\otimes\lambda_g\otimes\lambda_h,
\end{eqnarray}
where $\lambda_0$ denotes the identity operator $\mathbf{1}=|0\rangle\langle 0|+|1\rangle\langle 1|$, $\lambda_i$, $i=1,2,3$, represent the Pauli operators, i.e,  $\lambda_1=|0\rangle\langle 0|-|1\rangle\langle 1|, \lambda_2=|0\rangle\langle 1|+|1\rangle\langle 0|, \lambda_3=-i|0\rangle\langle 1|+i|1\rangle\langle 0|$ with the normalization $tr(\lambda_m\lambda_n)=2\delta_{mn}$ and $\langle\lambda_f\otimes\lambda_g\otimes\lambda_h\rangle=tr(\rho\lambda_f\otimes\lambda_g\otimes\lambda_h)$.
The correlation tensor $T$ of $\rho$ collects coefficients of $\lambda_f\otimes\lambda_g\otimes\lambda_h$, i.e, $\langle\lambda_f\otimes\lambda_g\otimes\lambda_h\rangle=t_{f,g,h}$, $1\leq f,g,h\leq 3$, and we arrange it in the following way $(T=(t_{f, gh})_{3\times 9})$:
\begin{eqnarray}
T=
\left( \begin{array}{ccccccccccccccc}
            t_{1,1,1}& t_{1,1,2}&t_{1,1,3}&t_{1,2,1}&t_{1,2,2}&t_{1,2,3}&t_{1,3,1}&t_{1,3,2}&t_{1,3,3}\\
            t_{2,1,1}& t_{2,1,2}&t_{2,1,3}&t_{2,2,1}&t_{2,2,2}&t_{2,2,3}&t_{2,3,1}&t_{2,3,2}&t_{2,3,3}\\
            t_{3,1,1}& t_{3,1,2}&t_{3,1,3}&t_{3,2,1}&t_{3,2,2}&t_{3,2,3}&t_{3,3,1}&t_{3,3,2}&t_{3,3,3}\\
           \end{array}
      \right ).
\label{A18}
\end{eqnarray}

Denote bipartitions of a tripartite quantum state $\rho$ as follows: $1|23$, $2|13$, and $3|12$. A tripartite quantum state $\rho$ is {\it biseparable} if $\rho=\sum_{i}p_i|\psi_i\rangle^{1|23}\langle\psi_i|+\sum_{j}p_j|\psi_j\rangle^{2|13}\langle\psi_j|+\sum_{k}p_k|\psi_k\rangle^{3|12}\langle\psi_k|$ with $p_i,p_j,p_k\geq0$ and $\sum_{i}p_i+\sum_{j}p_j+\sum_{k}p_k=1$. Otherwise, $\rho$ is called {\it genuinely multipartite entangled} (GME).
It follows from
$\mathbf{1}=|0\rangle\langle 0|+|1\rangle\langle 1|, \lambda_1=|0\rangle\langle 0|-|1\rangle\langle 1|, \lambda_2=|0\rangle\langle 1|+|1\rangle\langle 0|, \lambda_3=-i|0\rangle\langle 1|+i|1\rangle\langle 0|$ that
\begin{eqnarray}
\left\{
\begin{array}{lr}
|0\rangle\langle 0|=\frac12(\mathbf{1}+\lambda_1),\\
|1\rangle\langle 1|=\frac12(\mathbf{1}-\lambda_1),\\
|0\rangle\langle 1|=\frac12(\lambda_2+i\lambda_3),\\
|1\rangle\langle 0|=\frac12(\lambda_2-i\lambda_3).
 \end{array}
\right.
\end{eqnarray}

\begin{thm}\label{thm:1}
A tri-qubit mixed quantum state $\rho\in\mathcal{H}_1^2\otimes\mathcal{H}_2^2\otimes\mathcal{H}_3^2$ is GME, if
\begin{eqnarray}
\|T\|_{tr}>\sqrt{3}.
\end{eqnarray}
\end{thm}
\begin{proof}
First we consider a pure tri-qubit quantum state $\rho=|\psi\rangle\langle\psi|$. Assume $|\psi\rangle$ is $1|23$-separable, which can be viewed as a bipartite state, then its Schmidt decomposition [37] up to LU is $|\psi\rangle\overset{\text{LU}}{=}|0\rangle|a\rangle$, where $|a\rangle\in \mathcal{H}_1^2\otimes\mathcal{H}_2^2$. We divide it into 3 cases:

(1)
If $|a\rangle$ is separable, say $|a\rangle=|0\rangle|0\rangle$, then
\begin{eqnarray}
\rho=|\psi\rangle\langle\psi|=|000\rangle\langle 000|=\frac{1}{8}(\mathbf{1}+\lambda_1)\otimes(\mathbf{1}+\lambda_1)\otimes(\mathbf{1}+\lambda_1).
\end{eqnarray}
So we have $t_{111}=1$ and $\|T\|_{tr}=1$.

(2)
If $|a\rangle$ is inseparable and $|a\rangle=x|00\rangle+y|11\rangle$, then
\begin{eqnarray}
\rho=|\psi\rangle\langle\psi|=x^2|000\rangle\langle 000|+xy(|000\rangle\langle011|+|011\rangle\langle000|)+y^2|011\rangle\langle011|.
\end{eqnarray}
So we have $t_{111}=x^2+y^2=1$, $t_{122}=2xy$, and $t_{133}=-2xy$, thus $\|T\|_{tr}=\sqrt{1+8x^2y^2}\leq\sqrt{3}$.

(3)
If $|a\rangle$ is inseparable and $|a\rangle=x|01\rangle+y|10\rangle$, then
\begin{eqnarray}
\rho=|\psi\rangle\langle\psi|=x^2|001\rangle\langle 001|+xy(|001\rangle\langle010|+|010\rangle\langle001|)+y^2|010\rangle\langle010|.
\end{eqnarray}
So we have $t_{111}=-(x^2+y^2)=-1$, $t_{122}=2xy$, and $t_{133}=2xy$, thus $\|T\|_{tr}=\sqrt{1+8x^2y^2}\leq\sqrt{3}$.\

For $2|13$-separable and $3|12$-separable, we have the same result. Therefore, if a pure tri-qubit quantum state is biseparable, then $\|T\|_{tr}\leq\sqrt{3}$.

When $\rho$ is a tripartite biseparable mixed state, i.e, $\rho=\sum_{i}p_i|\psi_i\rangle^{1|23}\langle\psi_i|+\sum_{j}p_j|\psi_j\rangle^{2|13}\langle\psi_j|+\sum_{k}p_k|\psi_k\rangle^{3|12}\langle\psi_k|$
where $p_i,p_j,p_k\geq0$ and $\sum_{i}p_i+\sum_{j}p_j+\sum_{k}p_k=1$, then
\begin{eqnarray}
\|T(\rho)\|_{tr}&\leq&\sum_{i}p_i\|T(|\psi_i\rangle^{1|23})\|_{tr}+\sum_{j}p_j\|T(|\psi_j\rangle^{2|13})\|_{tr}+\sum_{k}p_k\|T(|\psi_k\rangle^{3|12})\|_{tr} \nonumber\\
&\leq& (\sum_{i}p_i+\sum_{j}p_j+\sum_{k}p_k)\sqrt{3} \nonumber\\
&=&\sqrt{3}.
\end{eqnarray}
\end{proof}
{\bf Remark 1.}
For a bipartite state $\rho\in\mathcal{H}_1^m\otimes\mathcal{H}_2^n$, if $\rho$ is separable, then $\|T\|_{tr}\leq\sqrt{\frac{mn(m-1)(n-1)}{4}}$ according to [29]. For a tri-qubit state $\rho\in\mathcal{H}_1^2\otimes\mathcal{H}_2^2\otimes\mathcal{H}_3^2$ under bipartition viewed as a bipartite state, one has when $\|T\|_{tr}>\sqrt{\frac{2\cdot4\cdot(2-1)(4-1)}{4}}=\sqrt{6}$, the tri-qubit state is GME. By Thm. 1, we obtain that for the tri-qubit quantum state, when $\|T\|_{tr}>\sqrt{3}$, the state is GME. Thus, Thm. 1 can detect more GME states than [29].

\textit{\textbf{Example 1.}} Let us consider a mixed state of GHZ state with white noise:
\begin{eqnarray}
\rho=\frac{x}{8}\mathbf{1}_8+(1-x)|\psi\rangle\langle\psi|,
\end{eqnarray}
where $|\psi\rangle=\frac{1}{\sqrt{2}}(|000\rangle+|111\rangle)$, $x\in[0,1]$ and $\mathbf{1}_8$ stands for the $8\times8$ identity matrix. One computes that $t_{2,2,2}=1-x$, $t_{3,3,2}=x-1$, $t_{3,2,3}=x-1$, and $t_{2,3,3}=x-1$. Then we have
\begin{eqnarray}
T=
\left( \begin{array}{ccccccccccccccc}
            0& 0&0&0&0&0&0&0&0\\
            0& 0&0&0&1-x&0&0&0&x-1\\
            0& 0&0&0&0&x-1&0&x-1&0\\
           \end{array}
      \right ).
\label{A2}
\end{eqnarray}
Thus $\|T\|_{tr}=2\sqrt{2}(1-x)$. Let $f_1(x)=\|T\|_{tr}-\sqrt{3}=2\sqrt{2}(1-x)-\sqrt{3}$, 
which is a monotonic decreasing function (see the red solid straight line in Fig. 1). By our Thm. 1, we have that when $f_1(x)>0$, the state $\rho$ is GME for $0\leq x<1-\frac{\sqrt{6}}{4}\approx0.38$. While \cite[Thm.1]{29} 
says that when $f_2(x)=2\sqrt{2}(1-x)-\sqrt{6}>0$, i.e, $0\leq x<1-\frac{\sqrt{3}}{2}\approx0.13$, the state is GME. In Ref. [34], $\rho$ is genuinely tripartite entangled state if $f_3(x)=2-2x-\sqrt{3}>0$, i.e, $0\leq x<0.13$. Fig. 1 depicts the comparison, which show that our Thm. 1 can detect more GME states than both Refs. [29, 34].
\begin{figure}[!htb]
\centerline{\includegraphics[width=15cm]{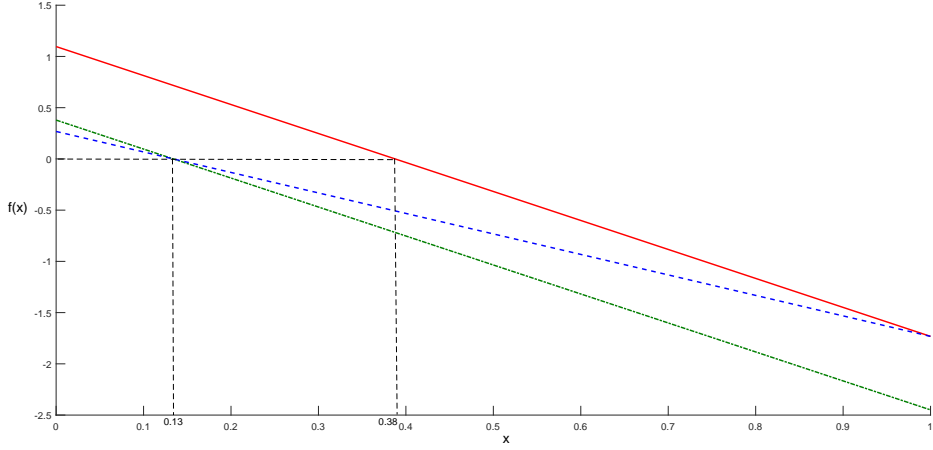}}
\renewcommand{\figurename}{Fig.}
\caption{
The red solid, green dashed and blue dashed straight lines represent $f_1(x)$ (from our Thm. 1), $f_2(x)$ (from [29]), and $f_3(x)$ (from [34]) respectively.}
\end{figure}

\textit{\textbf{Example 2.}} Next consider a generalized mixed W state,
\begin{eqnarray}
\rho=\frac{1-x}{8}\mathbf{1}_8+x|\varphi\rangle\langle\varphi|,
\end{eqnarray}
where $|\varphi\rangle=\frac{1}{\sqrt{3}}(|001\rangle+|010\rangle+|100\rangle)$, $x\in[0,1]$ and $\mathbf{1}_8$ stands for the $8\times8$ identity matrix.
By calculation, we have that $t_{111}=-x$, $t_{122}=\frac{2}{3}x$, $t_{212}=\frac{2}{3}x$, $t_{221}=\frac{2}{3}x$, $t_{133}=\frac{2}{3}x$, $t_{313}=\frac{2}{3}x$, $t_{331}=\frac{2}{3}x$, then
\begin{eqnarray}
T=
\left( \begin{array}{ccccccccccccccc}
            -x& 0&0&0&\frac{2}{3}x&0&0&0&\frac{2}{3}x\\
            0& \frac{2}{3}x&0&\frac{2}{3}x&0&0&0&0&0\\
            0& 0&\frac{2}{3}x&0&0&0&\frac{2}{3}x&0&0\\
           \end{array}
      \right ).
\label{A3}
\end{eqnarray}
Thus $\|T\|_{tr}=\frac{4\sqrt{2}+\sqrt{17}}{3}x$. It follows from Thm. 1 that when $\|T\|_{tr}>\sqrt{3}$, i.e, $0.5314\leq x\leq1$, the state $\rho$ is GME. In Ref. [38], $\rho$ is genuinely tripartite entangled state for $0.5464\leq x\leq1$. Therefore, our Thm. 1 can detect more GME states.

\section{GME for $d\otimes d\otimes d$ quantum states}
Now we consider genuine tripartite entanglement of $d\otimes d\otimes d$ quantum systems. Let $\lambda_m$, $m=1, \cdots d^2-1$ be the traceless Hermitian generators of $\mathfrak{su}(d)$, i.e,
\begin{eqnarray}
\lambda_m=\left\{
\begin{array}{lr}
\sqrt{\frac{2}{m(m+1)}}(\sum_{a=0}^{m-1}|a\rangle\langle a|-m|m\rangle\langle m|), \ \  m=1, \cdots,d-1;  \\
\lambda_{kj}=|k\rangle\langle j|+|j\rangle\langle k|,  \ \ 0\leq j<k\leq d-1 \ \ and \ \ m=d, \cdots, \frac{(d+2)(d-1)}2;\\
\lambda_{jk}=-i(|j\rangle\langle k|-|k\rangle\langle j|), \ \ 0\leq j<k\leq d-1\ \ and \ \  m=\frac{d(d+1)}2, \cdots, d^2-1,
 \end{array}
\right.
\end{eqnarray}
and $tr(\lambda_m\lambda_n)=d\delta_{mn}$.
Any tripartite quantum state $\rho\in\mathcal{H}_1^d\otimes\mathcal{H}_2^d\otimes\mathcal{H}_3^d$ can be represented as:
\begin{eqnarray}
\rho=\frac{1}{d^3}\sum_{f,g,h=0}^{d^2-1}\langle\lambda_f\otimes\lambda_g\otimes\lambda_h\rangle\lambda_f\otimes\lambda_g\otimes\lambda_h,
\end{eqnarray}
where $\lambda_0$ denotes the identity operator $\mathbf{1}$ and $\langle\lambda_f\otimes\lambda_g\otimes\lambda_h\rangle=tr(\rho\lambda_f\otimes\lambda_g\otimes\lambda_h)$.
Choosing the four generators $\lambda_{d-2}$, $\lambda_{d-1}$, $\lambda_{d-2\ d-1}$, and $\lambda_{d-1\ d-2}$, the correlation tensor $T$ of $\rho$ collects the coefficients of $\lambda_f\otimes\lambda_g\otimes\lambda_h$, i.e, $\langle\lambda_f\otimes\lambda_g\otimes\lambda_h\rangle=t_{f,g,h}$, $f,g,h\in\{d-2,d-1,d-2 \ d-1,d-1 \ d-2\}$. In the following, denote $d-1$ and $d-2$ by $-1$ and $-2$ respectively,  e.g., $t_{d-1,d-2,d-1 d-2}=t_{-1,-2,-1-2}$, $\lambda_{d-2\ d-1}=\lambda_{-2-1}$, $|d-1\rangle=|-1\rangle$. It is arranged as follows:

\begin{eqnarray}
T=
\left( \begin{array}{ccccccccccccccc}
            A_1&A_2&A_3&A_4
           \end{array}
      \right )
\label{A3b}
\end{eqnarray}
where
\begin{eqnarray}
A_1=
\left( \begin{array}{ccccccccccccccc}
            t_{-1,-1,-1}&t_{-1,-1,-2}&t_{-1,-1,-1-2}&t_{-1,-1,-2-1}\\
            t_{-2,-1,-1}&t_{-2,-1,-2}&t_{-2,-1,-1-2}&t_{-2,-1,-2-1}\\
            t_{-1-2,-1,-1}&t_{-1-2,-1,-2}&t_{-1-2,-1,-1-2}&t_{-1-2,-1,-2-1}\\
            t_{-2-1,-1,-1}&t_{-2-1,-1,-2}&t_{-2-1,-1,-1-2}&t_{-2-1,-1,-2-1}\\
           \end{array}
      \right ),
\label{A4}
\end{eqnarray}
\begin{eqnarray}
A_2=
\left( \begin{array}{ccccccccccccccc}
            t_{-1,-2,-1}&t_{-1,-2,-2}&t_{-1,-2,-1-2}&t_{-1,-2,-2-1}\\
            t_{-2,-2,-1}&t_{-2,-2,-2}&t_{-2,-2,-1-2}&t_{-2,-2,-2-1}\\
            t_{-1-2,-2,-1}&t_{-1-2,-2,-2}&t_{-1-2,-2,-1-2}&t_{-1-2,-2,-2-1}\\
            t_{-2-1,-2,-1}&t_{-2-1,-2,-2}&t_{-2-1,-2,-1-2}&t_{-2-1,-2,-2-1}\\
           \end{array}
      \right ),
\label{A5}
\end{eqnarray}
\begin{eqnarray}
A_3=
\left( \begin{array}{ccccccccccccccc}
            t_{-1,-1-2,-1}&t_{-1,-1-2,-2}&t_{-1,-1-2,-1-2}&t_{-1,-1-2,-2-1}\\
            t_{-2,-1-2,-1}&t_{-2,-1-2,-2}&t_{-2,-1-2,-1-2}&t_{-2,-1-2,-2-1}\\
            t_{-1-2,-1-2,-1}&t_{-1-2,-1-2,-2}&t_{-1-2,-1-2,-1-2}&t_{-1-2,-1-2,-2-1}\\
            t_{-2-1,-1-2,-1}&t_{-2-1,-1-2,-2}&t_{-2-1,-1-2,-1-2}&t_{-2-1,-1-2,-2-1}\\
           \end{array}
      \right ),
\label{A6}
\end{eqnarray}
\begin{eqnarray}
A_4=
\left( \begin{array}{ccccccccccccccc}
            t_{-1,-2-1,-1}&t_{-1,-2-1,-2}&t_{-1,-2-1,-1-2}&t_{-1,-2-1,-2-1}\\
            t_{-2,-2-1,-1}&t_{-2,-2-1,-2}&t_{-2,-2-1,-1-2}&t_{-2,-2-1,-2-1}\\
            t_{-1-2,-2-1,-1}&t_{-1-2,-2-1,-2}&t_{-1-2,-2-1,-1-2}&t_{-1-2,-2-1,-2-1}\\
            t_{-2-1,-2-1,-1}&t_{-2-1,-2-1,-2}&t_{-2-1,-2-1,-1-2}&t_{-2-1,-2-1,-2-1}\\
           \end{array}
      \right ).
\label{A7}
\end{eqnarray}

It follows from (14) that $\lambda_{-1}=\sqrt{\frac{2}{d(d-1)}}(\mathbf{1}-d|-1\rangle\langle-1|)$, $\lambda_{-2}=\sqrt{\frac{2}{(d-2)(d-1)}}(\mathbf{1}-|-1\rangle\langle-1|-(d-1)|-2\rangle\langle-2|)$, $\lambda_{-1-2}=|-1\rangle\langle-2|+|-2\rangle\langle-1|$, and $\lambda_{-2-1}=-i|-2\rangle\langle-1|+i|-1\rangle\langle-2|$. Thus,
\begin{eqnarray}
|-1\rangle\langle-1|&=&\frac{1}{d}(\mathbf{1}-\sqrt{\frac{d(d-1)}{2}}\lambda_{-1}),\\
|-2\rangle\langle-2|&=&\frac{1}{d}\mathbf{1}+\sqrt{\frac{1}{2d(d-1)}}\lambda_{-1}-\sqrt{\frac{d-2}{2(d-1)}}\lambda_{-2},\\
|-1\rangle\langle-2|&=&\frac{1}{2}(\lambda_{-1-2}-i\lambda_{-2-1}),\\
|-2\rangle\langle-1|&=&\frac{1}{2}(\lambda_{-1-2}+i\lambda_{-2-1}).
\end{eqnarray}

\begin{thm}\label{thm:2}
A tripartite mixed quantum state $\rho\in\mathcal{H}_1^d\otimes\mathcal{H}_2^d\otimes\mathcal{H}_3^d$ is GME, if
\begin{eqnarray}
\|T\|_{tr}>\sqrt{\frac{d^3(d-1)(d^2-d+1)}{8}}.
\end{eqnarray}
\end{thm}
\begin{proof}
Enough to consider a pure tripartite quantum state $\rho=|\psi\rangle\langle\psi|\in\mathcal{H}_1^d\otimes\mathcal{H}_2^d\otimes\mathcal{H}_3^d$. Assume $|\psi\rangle$ is $1|23$-separable, its Schmidt decomposition is $|\psi\rangle\overset{\text{LU}}{=}|-1\rangle|a\rangle$, where $|a\rangle\in \mathcal{H}_2^{d}\otimes\mathcal{H}_3^{d}$. We discuss as follows:

(1) If $|a\rangle$ is separable, i.e, $|a\rangle=|-1\rangle|-1\rangle$, then
\begin{eqnarray}
\rho=|\psi\rangle\langle\psi|&=&|-1,-1,-1\rangle\langle -1,-1,-1|\nonumber\\
&=&\frac{1}{d^3}(\mathbf{1}-\sqrt{\frac{d(d-1)}{2}}\lambda_{-1})\otimes(\mathbf{1}-\sqrt{\frac{d(d-1)}{2}}\lambda_{-1})\otimes(\mathbf{1}-\sqrt{\frac{d(d-1)}{2}}\lambda_{-1}).
\end{eqnarray}
So we have $t_{-1,-1,-1}=tr(\rho\lambda_{-1}\otimes\lambda_{-1}\otimes\lambda_{-1})=-\frac{\sqrt{2d^3(d-1)^3}}{4}$, thus $\|T\|_{tr}=1$.

(2) If $|a\rangle$ is nonseparable and $|a\rangle=x|-1, -1\rangle+y|-2, -2\rangle$, then
\begin{eqnarray}
\rho=|\psi\rangle\langle\psi|&=&x^2|-1,-1,-1\rangle\langle-1,-1,-1|+xy(|-1,-1,-1\rangle\langle-1,-2,-2|\nonumber\\
&+&|-1,-2,-2\rangle\langle-1,-1,-1|)+y^2|-1,-2,-2\rangle\langle-1,-2,-2|.
\end{eqnarray}
So we have $t_{-1,-1,-1}=tr(\rho\lambda_{-1}\otimes\lambda_{-1}\otimes\lambda_{-1})=\frac{d^2\sqrt{d-1}}{2\sqrt{2}}[-(d-1)\sqrt{\frac{1}{d}}x^2-\frac1{d-1}\sqrt{\frac{1}{d}}y^2]$, $t_{-1,-1,-2}=tr(\rho\lambda_{-1}\otimes\lambda_{-1}\otimes\lambda_{-2})=\frac{d^2\sqrt{d-1}}{2\sqrt{2}}(\frac{\sqrt{d-2}}{d-1}y^2)$, $t_{-1,-2,-1}=tr(\rho\lambda_{-1}\otimes\lambda_{-2}\otimes\lambda_{-1})=\frac{d^2\sqrt{d-1}}{2\sqrt{2}}(\frac{\sqrt{d-2}}{d-1}y^2)$,
$t_{-1,-2,-2}=tr(\rho\lambda_{-1}\otimes\lambda_{-2}\otimes\lambda_{-1})=\frac{d^2\sqrt{d-1}}{2\sqrt{2}}(-\frac{d-2}{d-1}\sqrt{d}y^2)$,
$t_{-1,-1-2,-1-2}=tr(\rho\lambda_{-1}\otimes\lambda_{-1-2}\otimes\lambda_{-1-2})=\frac{d^2\sqrt{d-1}}{2\sqrt{2}}(-\sqrt{d}xy)$, $t_{-1,-2-1,-2-1}=tr(\rho\lambda_{-1}\otimes\lambda_{-2-1}\otimes\lambda_{-2-1})=\frac{d^2\sqrt{d-1}}{2\sqrt{2}}(\sqrt{d}xy)$, then
\begin{eqnarray}
T=
\left( \begin{array}{ccccccccccccccc}
            A_1&A_2&A_3&A_4
           \end{array}
      \right )
\label{A8}
\end{eqnarray}
where
\begin{eqnarray}
A_1=\frac{d^2\sqrt{d-1}}{2\sqrt{2}}
\left( \begin{array}{ccccccccccccccc}
            -(d-1)\sqrt{\frac{1}{d}}x^2-\frac{1}{d-1}\sqrt{\frac{1}{d}}y^2&\frac{1}{d-1}\sqrt{d-2}y^2&0&0\\
            0&0&0&0\\
            0&0&0&0\\
            0&0&0&0\\
           \end{array}
      \right ),
\label{A9}
\end{eqnarray}
\begin{eqnarray}
A_2=\frac{d^2\sqrt{d-1}}{2\sqrt{2}}
\left( \begin{array}{ccccccccccccccc}
            \frac{1}{d-1}\sqrt{d-2}y^2&-\frac{d-2}{d-1}\sqrt{d}y^2&0&0\\
            0&0&0&0\\
            0&0&0&0\\
            0&0&0&0\\
           \end{array}
      \right ),
\label{A10}
\end{eqnarray}
\begin{eqnarray}
A_3=\frac{d^2\sqrt{d-1}}{2\sqrt{2}}
\left( \begin{array}{ccccccccccccccc}
            0&0&-\sqrt{d}xy&0\\
            0&0&0&0\\
            0&0&0&0\\
            0&0&0&0\\
           \end{array}
      \right ),
\label{A11}
A_4=\frac{d^2\sqrt{d-1}}{2\sqrt{2}}
\left( \begin{array}{ccccccccccccccc}
            0&0&0&\sqrt{d}xy\\
            0&0&0&0\\
            0&0&0&0\\
            0&0&0&0\\
           \end{array}
      \right ).
\label{A12}
\end{eqnarray}
Therefore,
\begin{eqnarray}
\|T\|_{tr}&=&Tr(\sqrt{A_1A_1^T+A_1A_2^T+A_3A_3^T+A_4A_4^T})\nonumber\\
&=&\frac{d^2\sqrt{d-1}}{2\sqrt{2}}\sqrt{\frac{(d-1)^2}{d}(x^4+y^4)+(\frac{2}{d}+2d)x^2y^2}\nonumber\\
&=&\frac{d^2\sqrt{d-1}}{2\sqrt{2}}\sqrt{\frac{(d-1)^2}{d}[(x^2+y^2)^2-2x^2y^2]+(\frac{2}{d}+2d)x^2y^2}\nonumber\\
&=&\frac{d^2\sqrt{d-1}}{2\sqrt{2}}\sqrt{\frac{1}{d}+d-2+4x^2y^2}\nonumber\\
&\leq&\frac{d^2\sqrt{d-1}}{2\sqrt{2}}\sqrt{\frac{1}{d}+d-2+4(\frac{1}{2})^2} \nonumber\\
&=&\sqrt{\frac{d^3(d-1)(d^2-d+1)}{8}},
\end{eqnarray}
where the inequality is due to $xy\leq\frac{x^2+y^2}{2}=\frac{1}{2}$.

(3) If $|a\rangle$ is inseparable and $|a\rangle=x|-1, -2\rangle+y|-2, -1\rangle$, then
\begin{eqnarray}
\rho=|\psi\rangle\langle\psi|&=&x^2|-1,-1,-2\rangle\langle-1,-1,-2|+xy(|-1,-1,-2\rangle\langle-1,-2,-1|\nonumber\\
&+&|-1,-2,-1\rangle\langle-1,-1,-2|)+y^2|-1,-2,-1\rangle\langle-1,-2,-1|.
\end{eqnarray}
So we have $t_{-1,-1,-1}=tr(\rho\lambda_{-1}\otimes\lambda_{-1}\otimes\lambda_{-1})=\frac{d^2\sqrt{d-1}}{2\sqrt{2}}\sqrt{\frac{1}{d}}$, $t_{-1,-1,-2}=tr(\rho\lambda_{-1}\otimes\lambda_{-1}\otimes\lambda_{-2})=\frac{d^2\sqrt{d-1}}{2\sqrt{2}}(-\sqrt{d-2}x^2)$, $t_{-1,-2,-1}=tr(\rho\lambda_{-1}\otimes\lambda_{-2}\otimes\lambda_{-1})=\frac{d^2\sqrt{d-1}}{2\sqrt{2}}(-\sqrt{d-2}y^2)$, $t_{-1,-1-2,-1-2}=tr(\rho\lambda_{-1}\otimes\lambda_{-1-2}\otimes\lambda_{-1-2})=\frac{d^2\sqrt{d-1}}{2\sqrt{2}}(-\sqrt{d}xy)$, $t_{-1,-2-1,-2-1}=tr(\rho\lambda_{-1}\otimes\lambda_{-2-1}\otimes\lambda_{-2-1})=\frac{d^2\sqrt{d-1}}{2\sqrt{2}}(-\sqrt{d}xy)$, then
\begin{eqnarray}
T=
\left( \begin{array}{ccccccccccccccc}
            A_1&A_2&A_3&A_4
           \end{array}
      \right )
\label{A13}
\end{eqnarray}
where
\begin{eqnarray}
A_1=\frac{d^2\sqrt{d-1}}{2\sqrt{2}}
\left( \begin{array}{ccccccccccccccc}
            \sqrt{\frac{1}{d}}&-\sqrt{d-2}x^2&0&0\\
            0&0&0&0\\
            0&0&0&0\\
            0&0&0&0\\
           \end{array}
      \right ),
\label{A14}
A_2=\frac{d^2\sqrt{d-1}}{2\sqrt{2}}
\left( \begin{array}{ccccccccccccccc}
            -\sqrt{d-2}y^2&0&0&0\\
            0&0&0&0\\
            0&0&0&0\\
            0&0&0&0\\
           \end{array}
      \right ),
\label{A15}
\end{eqnarray}
\begin{eqnarray}
A_3=\frac{d^2\sqrt{d-1}}{2\sqrt{2}}
\left( \begin{array}{ccccccccccccccc}
            0&0&-\sqrt{d}xy&0\\
            0&0&0&0\\
            0&0&0&0\\
            0&0&0&0\\
           \end{array}
      \right ),
\label{A16}
A_4=\frac{d^2\sqrt{d-1}}{2\sqrt{2}}
\left( \begin{array}{ccccccccccccccc}
            0&0&0&-\sqrt{d}xy\\
            0&0&0&0\\
            0&0&0&0\\
            0&0&0&0\\
           \end{array}
      \right ).
\label{A17}
\end{eqnarray}
Therefore,
\begin{eqnarray}
\|T\|_{tr}&=&Tr(\sqrt{A_1A_1^T+A_1A_2^T+A_3A_3^T+A_4A_4^T})\nonumber\\
&=&\frac{d^2\sqrt{d-1}}{2\sqrt{2}}\sqrt{\frac{1}{d}+(d-2)(x^4+y^4)+2dx^2y^2}\nonumber\\
&=&\frac{d^2\sqrt{d-1}}{2\sqrt{2}}\sqrt{\frac{1}{d}+(d-2)[(x^2+y^2)^2-2x^2y^2]+2dx^2y^2}\nonumber\\
&=&\frac{d^2\sqrt{d-1}}{2\sqrt{2}}\sqrt{\frac{1}{d}+d-2+4x^2y^2}\nonumber\\
&\leq&\frac{d^2\sqrt{d-1}}{2\sqrt{2}}\sqrt{\frac{1}{d}+d-2+4(\frac{1}{2})^2} \nonumber\\
&=&\sqrt{\frac{d^3(d-1)(d^2-d+1)}{8}},
\end{eqnarray}
where the inequality is due to $xy\leq\frac{x^2+y^2}{2}=\frac{1}{2}$.
\end{proof}
{\bf Remark 2.} When d=2, the theorem says if $\|T\|_{tr}>\sqrt{\frac{d^3(d-1)(d^2-d+1)}{8}}=\sqrt{3}$, then the tripartite quantum state $\rho$ is GME. Thus, Thm. 2 is a generalization of Thm. 1.

{\bf Remark 3.} For a tri-qubit state $\rho\in\mathcal{H}_1^d\otimes\mathcal{H}_2^d\otimes\mathcal{H}_3^d$ under bipartition viewed as a bipartite state, one has if $\|T\|_{tr}>\sqrt{\frac{d^3(d^2-1)(d-1)}{4}}$, then the tri-qubit state is genuinely multipartite entangled by Ref. [29].
Since $\sqrt{\frac{d^3(d-1)(d^2-d+1)}{8}}-\sqrt{\frac{d^3(d^2-1)(d-1)}{4}}=\sqrt{\frac{d^3(d-1)}{4}}(\sqrt{\frac{d^2-d+1}{2}}-\sqrt{d^2-1})<0$ for $d\geq2$, Thm. 2 can detect more GME states than Ref. [29].

\textit{\textbf{Example 3.}} Consider the 3-qutrit state
\begin{eqnarray}
\rho=\frac{x}{27}\mathbf{1}_{27}+(1-x)|GGHZ\rangle\langle GGHZ|,
\end{eqnarray}
where $|GGHZ\rangle=\frac{1}{\sqrt{3}}(|000\rangle+|111\rangle+|222\rangle)$ is a generalized GHZ state and $x\in[0,1]$. For d=3, 
the tripartite quantum state $\rho$ is GME if $\|T\|_{tr}>\sqrt{\frac{d^3(d-1)(d^2-d+1)}{8}}=\frac{3\sqrt{21}}{2}$ by Thm. 2 or $\|T\|_{tr}>\sqrt{\frac{d^3(d^2-1)(d-1)}{4}}=6\sqrt{3}$ by Ref. [29]. One computes that $t_{10,10,10}=\frac{9(1-x)}{4}$, $t_{10,01,01}=-\frac{9(1-x)}{4}$, $t_{01,01,10}=-\frac{9(1-x)}{4}$, $t_{01,10,01}=-\frac{9(1-x)}{4}$, $t_{20,20,20}=\frac{9(1-x)}{4}$, $t_{20,02,02}=-\frac{9(1-x)}{4}$, $t_{02,02,20}=-\frac{9(1-x)}{4}$, $t_{02,20,02}=-\frac{9(1-x)}{4}$, $t_{2,2,2}=-\frac{3\sqrt{3}(1-x)}{4}$, $t_{1,1,2}=\frac{3\sqrt{3}(1-x)}{4}$, $t_{1,2,1}=\frac{3\sqrt{3}(1-x)}{4}$, $t_{2,1,1}=\frac{3\sqrt{3}(1-x)}{4}$, $t_{21,21,21}=\frac{9(1-x)}{4}$, $t_{21,12,12}=-\frac{9(1-x)}{4}$, $t_{12,21,12}=-\frac{9(1-x)}{4}$, and $t_{12,12,21}=-\frac{9(1-x)}{4}$.
Thus $\|T\|_{tr}=\frac{27\sqrt{2}+3\sqrt{6}}{2}(1-x)$. So when $0\leq x<0.69$, the state $\rho$ is GME by Thm. 2. However by Ref. [29] $\rho$ is GME
when $0\leq x<0.54$. Therefore, Theorem 2 can detect more GME states than Ref. [29].

\section{Conclusions}
In this article, we simplify the study of genuine multipartite entanglement by using local unitary equivalence and the correlation tensor. With the simplified argument
we have obtained stronger GME criteria for tri-qubit quantum states and tripartite higher-dimensional quantum states.
Using examples we have shown that our criteria outperform some of the well-known ones in detecting GME.

\bigskip

\textbf {Acknowledgements}
This work is supported in part by Simons Foundation under grant no. 523868 and NSFC under grant nos. 12126351 and 12126314.

\bigskip

\textbf{Data Availability Statement.} All data generated during the study are included in the article.

\end{CJK*}

\begin{thebibliography}{99}
\bibitem{1} R. Horodecki, P. Horodecki, M. Horodecki, and K. Horodecki, Quantum entanglement. Rev. Mod. Phys. 81 (2), 865 (2009).
\bibitem{2} A. K. Ekert, Quantum cryptography based on Bell's theorem. Phys. Rev. Lett. 67 (6), 661 (1991).
\bibitem{3} D. Goyeneche, D. Alsina, J. I. Latorre, A. Riera, and K. \.{Z}yczkowski, Absolutely maximally entangled states, combinatorial designs and multiunitary matrices. Phys. Rev. A 92 (3), 032316 (2015).
\bibitem{4} L. Gyongyosi and S. Imre, Entanglement access control for the quantum internet. Quantum Inf. Process. 18 (4), 107 (2019).
\bibitem{5} S. K. Goyal, S. Banerjee, and S. Ghosh, Effect of control procedures on the evolution of entanglement in open quantum systems. Phys. Rev. A 85 (1), 012327 (2012).
\bibitem{6} J. D. Bancal, N. Gisin, Y. C. Liang, and S. Pironio, Device-independent witnesses of genuine multipartite entanglement. Phys. Rev. Lett. 106 (25), 250404 (2011).
\bibitem{7} M. Seevinck, and J. Uffink, Sufficient conditions for three-particle entanglement and their tests in recent experiments. Phys. Rev. A 65 (1), 012107 (2002).
\bibitem{8} M. Seevinck, and G. Svetlichny, Bell-type inequalities for partial separability in N-particle systems and quantum mechanical violations. Phys. Rev. Lett. 89 (6), 060401 (2002).
\bibitem{9} M. Horodecki, P. Horodecki, and R. Horodecki, Separability of mixed states: necessary and sufficient conditions. Phys. Lett. A 223 (1), 1 (1996).
\bibitem{10} B. M. Terhal, Bell inequalities and the separability criterion. Phys. Lett. A 271 (5), 319 (2000).
\bibitem{11} B. Jungnitsch, T. Moroder, and O. G\"{u}hne, Entanglement witnesses for graph states: general theory and examples. Phys. Rev. A 84 (3), 032310 (2011).
\bibitem{12} J. Y. Wu, H. Kampermann, D. Bru\ss, C. Kl\"{o}ckl, and M. Huber, Determining lower bounds on a measure of multipartite entanglement from few local observables. Phys. Rev. A 86 (2), 022319 (2012).
\bibitem{13} J. Sperling, and W. Vogel, Multipartite entanglement witnesses. Phys. Rev. Lett. 111 (11), 110503 (2013).
\bibitem{14} C. Eltschka, and J. Siewert, Quantifying entanglement resources. J. Phys. A : Math. Theor. 47 (42), 424005 (2014).
\bibitem{15} Y. Akbari-Kourbolagh and M. Azhdargalam, Entanglement criterion for multipartite systems based on quantum Fisher information. Phys. Rev. A 99 (1), 012304 (2019).
\bibitem{16} Z. H. Ma, Z. H. Chen, J. L. Chen, C. Spengler, A. Gabriel, and M. Huber, Measure of genuine multipartite entanglement with computable lower bounds. Phys. Rev. A 83 (6), 062325 (2011).
\bibitem{17} Z. H. Chen, Z. H. Ma, J. L. Chen, and S. Severini, Improved lower bounds on genuine-multipartite-entanglement concurrence. Phys. Rev. A 85 (6), 062320 (2012).
\bibitem{18} Y. Hong, T. Gao, and F. L. Yan, Measure of multipartite entanglement with computable lower bounds. Phys. Rev. A 86 (6), 062323 (2012).
\bibitem{19} T. Gao, F. L. Yan, and S. J. van Enk, Permutationally invariant part of a density matrix and nonseparability of N-qubit states. Phys. Rev. Lett. 112 (18), 180501 (2014).
\bibitem{20} H. F. Hofmann and S. Takeuchi, Violation of local uncertainty relations as a signature of entanglement. Phys. Rev. A 68 (3), 032103 (2003).
\bibitem{21} C. J. Zhang, H. Nha, Y. S. Zhang, and G. C. Guo, Entanglement detection via tighter local uncertainty relations. Phys. Rev. A 81 (1), 012324 (2010).
\bibitem{22} Y. Akbari-Kourbolagh, and M. Azhdargalam, Entanglement criterion for tripartite systems based on local sum uncertainty relations. Phys. Rev. A 97 (4), 042333 (2018).
\bibitem{23} J. Li, and L. Chen, Detection of genuine multipartite entanglement based on uncertainty relations. Quantum Inf. Process. 20 (6), 220 (2021).
\bibitem{24} J. I. de Vicente, Further results on entanglement detection and quantification from the correlation matrix criterion. J. Phys. A 41, 065309 (2008).
\bibitem{25} A. S. M. Hassan, and P. S. Joag, Separability criterion for multipartite quantum states based on the bloch representation of density matrices. Quantum Inf. Comput. 8 (8), 773 (2008).
\bibitem{26} M. Li, J. Wang, S. M. Fei, and X. Li-Jost, Quantum separability criteria for arbitrary dimensional multipartite states. Phys. Rev. A 89 (2), 022325 (2014).
\bibitem{27} S. Q. Shen, J. Yu, M. Li, and S. M. Fei, Improved separability criteria based on Bloch representation of density matrices. Sci. Rep. 6, 28850 (2016).
\bibitem{28} M. Li, Z. Wang, J. Wang, S. Shen, and S. M. Fei, The norms of Bloch vectors and classification of four-qudits quantum states. Europhys. Lett. A 125 (2), 20006 (2019).
\bibitem{29} J. I. de Vicente, Separability criteria based on the Bloch representation of density matrices. Quantum Inf. Comput. 7 (7), 624, (2007).
\bibitem{30} M. Li, L. Jia, J. Wang, S. Shen, and S. M. Fei, Measure and detection of genuine multipartite entanglement for tripartite systems. Phys. Rev. A 96 (5), 052314 (2017).
\bibitem{31} H. Zhao, Y. Q. Liu, N. Jing, Z. X. Wang, and S. M. Fei, Detection of genuine tripartite entanglement based on Bloch representation of density matrices. Quantum Inf. Process. 21 (3), 116 (2022).
\bibitem{32} A. Elben, R. Kueng, H. Y. Huang, et al., Mixed-state entanglement from local randomized measurements. Phys. Rev. Lett. 125 (20), 200501 (2020).
\bibitem{33} X. D. Yu, S. Imai, and O. G\"{u}hne, Optimal entanglement certification from moments of the partial transpose. Phys. Rev. Lett. 127 (6), 060504 (2021).
\bibitem{34} H. Zhao, L. Liu, Z. X. Wang, N. Jing, and J. Li, On genuine entanglement for tripartite systems. Int. J. Quantum Inf. 20(2), 2150038 (2022).
\bibitem{35} N. Jing, M. Yang, and H. Zhao, Local unitary equivalence of quantum states and simultaneous orthogonal equivalence. J. Math. Phys. 57 (6), 062205 (2016).
\bibitem{36} N. Jing, S. M. Fei, M. Li, X. Li-Jost, and T. Zhang, Local unitary invariants of generic multiqubit states. Phys. Rev. A 92 (2), 022306 (2015).
\bibitem{37} M. A. Nielsen, and I. L. Chuang, {\it Quantum Computation and Quantum Information} (Cambridge University Press, Cambridge, England), 2000.
\bibitem{38} H. Zhao, Y. Q. Liu, N. Jing, and Z. X. Wang, Detection of genuine entanglement for multipartite quantum states. Quantum Inf. Process. 21 (9), 315 (2022).



\end{thebibliography}
\end{document}